\RequirePackage{amsthm}

\let\LaTeXcline\cline 
\documentclass[pdflatex,sn-mathphys-num]{sn-jnl}
\let\cline\LaTeXcline

\usepackage{graphicx}%
\usepackage{multirow}%
\usepackage{amsmath,amssymb,amsfonts}%
\usepackage{amsthm}%
\usepackage{mathrsfs}%
\usepackage[title]{appendix}%
\usepackage{xcolor}%
\usepackage{textcomp}%
\usepackage{manyfoot}%
\usepackage{booktabs}%
\usepackage{algorithmicx}%
\usepackage{algpseudocode}%
\usepackage{listings}%

\usepackage{mathtools} 

\usepackage[ruled,norelsize,linesnumbered,procnumbered]{algorithm2e}

\makeatletter
\AtBeginEnvironment{procedure}{\let\c@algocf\c@procedure}
\makeatother

\newcommand{\be}{\begin{equation}}
\newcommand{\ee}{\end{equation}}

\newcommand{\bald}{\begin{aligned}}
\newcommand{\eald}{\end{aligned}}

\newcommand{\bbm}{\begin{bmatrix}}
\newcommand{\ebm}{\end{bmatrix}}

\newcommand{\bxi}{\boldsymbol{\xi}}

\newcommand{\bx}{\boldsymbol{x}}

\newcommand{\R}{\mathbb{R}}

\newcommand{\bb}{\boldsymbol{b}}

\newcommand{\bz}{\boldsymbol{z}}

\newcommand{\bW}{\boldsymbol{W}}

\newtheorem{rem}{Remark}
\newtheorem{prop}{Proposition}

\newcommand{\mc}[1]{\mathcal{#1}}

\newcommand{\bA}{\mathbf{A}}
\newcommand{\bB}{\mathbf{B}}
\newcommand{\bpf}{\mathbf{p}}
\newcommand{\bPf}{\mathbf{P}}
\newcommand{\bQf}{\mathbf{Q}}
\newcommand{\bRf}{\mathbf{R}}

\newcommand{\bK}{\boldsymbol{K}}

\newcommand{\bzeta}{\boldsymbol{\zeta}}

\newcommand{\balpha}{\boldsymbol{\alpha}}

\newcommand{\bbms}{\begin{bsmallmatrix}}
\newcommand{\ebms}{\end{bsmallmatrix}}

\definecolor{Modcolor}{HTML}{000000}
\definecolor{Updatecolor}{HTML}{000000}

\newcommand{\thinh}[1]{\textcolor{Modcolor}{#1}}

\usepackage{tikz}
\usepackage{color}
\usepackage{multirow}
\usepackage{amssymb}
\usepackage{gensymb}
\usepackage{tabularx}
\usepackage{extarrows}
\usetikzlibrary{fadings}
\usetikzlibrary{patterns}
\usetikzlibrary{shadows.blur}
\usetikzlibrary{shapes}

\usepackage{pgfplots}
\pgfplotsset{width=10cm,compat=1.9}

\definecolor{midnightGreen}{RGB}{11,85,99}
\definecolor{celestialBlue}{RGB}{82, 153, 211}

\usepackage{cases} 

\usepackage{natbib}
\newtheorem{coll}{Corollary}
\hypersetup{
    pdftitle={Positive Invariance of flat systems via ANN}
    }

\theoremstyle{thmstyleone}%
\theoremstyle{thmstyletwo}%

\theoremstyle{thmstylethree}%
\newtheorem{definition}{Definition}%

\begin{document}

\title[ANN-based Invariance for flat systems]{On ANN-enhanced positive invariance for nonlinear flat systems}


\author*[1,2]{\fnm{Huu-Thinh} \sur{Do}}
\equalcont{\small{ This work was carried out during the postdoctoral research of Huu-Thinh Do at Univ. Grenoble Alpes, Laboratoire de Conception et d'Intégration des Systèmes (LCIS), Valence, France.}}

\author[1]{\fnm{Ionela} \sur{Prodan}}



\affil[1]{\orgname{Univ. Grenoble Alpes, Grenoble INP$^\ddag$, LCIS} \orgaddress{, \city{Valence}, \postcode{F-26000}, \country{France.\\$\ddag$Institute of Engineering and Management University Grenoble Alpes}}}

\affil[2]{\orgname{Department of Aerospace Engineering, University of Michigan}, \orgaddress{ \city{Ann Arbor}, \postcode{48109}, \country{USA}}}




\abstract{
The concept of positively invariant (PI) sets has proven effective in the formal verification of stability and safety properties for autonomous systems. However, the characterization of such sets is challenging for nonlinear systems in general, especially in the presence of constraints. In this work, we show that, for a class of feedback linearizable systems, called differentially flat systems, a PI set can be derived by leveraging a neural network approximation of the linearizing mapping. More specifically, for the class of flat systems, there exists a linearizing variable transformation that converts the nonlinear system into linear  controllable dynamics, albeit at the cost of distorting the constraint set. We show that by approximating the distorted set using a rectified linear unit neural network, we can derive a PI set inside the admissible domain through  its set-theoretic description. This offline characterization enables the synthesis of various efficient online control strategies, with different complexities and performances. Numerical simulations are provided to demonstrate the validity of the proposed framework.
}




\maketitle

\section{Introduction}
\label{sec:Intro}


\thinh{Handling constraints and verifying stability have always been critical issues in both system and control theory \cite{mayne2000constrained,liu2017dynamical,wang2020variational,li2017nonlocal}.}
For example, for robotic applications, this problem has been regarded as keeping the control system inside a so-called {safe set} \citep{long2025safety,liu2017designing,compton2024dynamic} where safety indices are maintained with stability guaranteed. From a control theoretic viewpoint, such a definition coincides with the notion of the \textit{positively invariant} (PI) set \citep{blanchini1999set}, a subset of the operating space within which the trajectory of the system, once it enters, remains for all future time and continuously satisfies the constraints. As a result, the existence of a PI set ensures the feasibility of the stabilization problem.
Typically, the construction of the set is associated with a positive definite Control Lyapunov Function (CLF) decreasing inside the PI set along the system's trajectory. However, computing a function satisfying such a condition is challenging for nonlinear systems in general. For more details, we refer to the systematic surveys \citep{giesl2015review,dawson2023safe}.

In contrast, finding a PI set with a CLF for constrained linear systems are rather well developed in the literature with efficient theoretical and numerical tools (see e.g., \citep{blanchini1999set,boyd1993control} and the references therein). Thus, to design PI sets for nonlinear systems, linearization is a common first step. 
There exists two main types of linearizations, the first approach relies on the first-order Taylor expansion of the dynamics near the equilibrium point. This direction provides a general procedure but is typically challenging to quantify the approximation error, and hence the validity of the PI set with respect to the original nonlinear dynamics. 
The second approach, called feedback linearization, seeks a variable transformation that simplifies the dynamics into linear ordinary differential equations \citep{soydemir2024online,zheng2022extreme}.
Within this framework, the theory of differentially flat systems \citep{fliess1995flatness} emerges as an efficient tool, demonstrating that for a specific class of systems, the linearization mapping constitutes an equivalence transformation, thereby ensuring that the stability properties of the system in the transformed (called the flat output space) and original space are equivalent. 

In this manner, the dynamical properties of the systems are preserved while characterizing the invariance. The main difficulty of this problem, however, comes with the presence of constraints. This arises from the fact that, in the new coordinates, the description of the constraints will be distorted, complicating the stability design in convex optimization-based frameworks \citep{kandler2012differential,gratzer2024two}.


\begin{figure}[thp]
    \centering
    \includegraphics[width=0.675\linewidth]{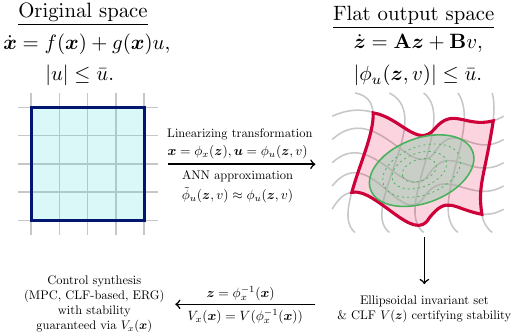}
    \caption{Searching for ellipsoidal PI set in the flat output space for control design in the original space.}
    \label{fig:IntroIllus}
\end{figure}

To sidestep the problem, various approaches have been proposed to obtain an approximation of the convoluted constraints, or the linearizing variable change, transforming the control synthesis into stabilizing a linear systems with manageable constraints. For instance, the studies \citep{hall2023differentially,greeff2020exploiting} proposed to approximate the linearizing mapping with Gaussian processes, embedding invariance certificate into a linear setting with a standard quadratic CLF.
\thinh{However, such an approach offers only probabilistic guarantees, which makes the resulting invariance property inherently non-deterministic.}
From a set-theoretic setting, by introducing a hyperbox subset of the distorted constraints, it was proven in \citep{nguyen2024notes} that, one can construct for the linearized dynamics, under the form of three double integrators, an arbitrarily large ellipsoidal PI set with, again, a quadratic CLF. Similarly, in \citep{bozza2024online}, by linearly approximating the dynamics over a set basis functions, the stability is guaranteed for the feedback linearization scheme in a data-driven fashion, via convex programming.

Aligned with this orientation, our prior works \citep{DoMILP2024,thinhPWAANN} explores the efficacy of Rectified linear unit (ReLU) artificial neural networks (ANN) to address the constrained control of differentially flat (flat) or feedback linearizable systems. Therein, by approximating the linearizing mapping with ReLU-ANN, we showed that the distorted constraints can be taken into account via an online mixed-integer program (MIP).
This is thanks to the piecewise affine nature of the ReLU activation function and its versatility \citep{grimstad2019relu}.
However, with this direction, real-time solvability is not guaranteed as MIP is generally NP-hard. Thus, in this work, with a view to shifting the computational burden offline, we:
\begin{itemize}
\item demonstrate that the approximated constraints can be represented as a union of polytopes, laying the ground work for set-based invariance characterization;
    \item propose an offline procedure to construct a PI set associated with a CLF which is quadratic in the linearizing space (see Fig. \ref{fig:IntroIllus}), \thinh{providing a formal verification for the stability and constraint satisfaction};
    \item synthesize different controllers based on the detected PI set, with different levels of online complexity and performance, from CLF-based control, Model Predictive Control (MPC) to Explicit Reference Governor (ERG). 
\end{itemize}











The remainder of the paper is organized as follows. Section \ref{sec:preliminaries} provides the preliminaries on the context of controlling flat systems in the linearizing coordinate with the distorted constraint, followed by a set-theoretic description of its ANN approximation. With the linear representation in the new space, this description then gives 
rise to the invariance characterization in Section \ref{sec:Invariancecharacterization}. The utility of the PI set detected is then demonstrated via the control synthesis in optimization-based framework, including CLF-based control, MPC, and ERG. Section \ref{sec:IllusExam} presents numerical simulations to validate the proposed setting. Finally, section \ref{sec:concl} concludes and discusses future directions.

\textit{Notation:} Bold lower-case letters denote vector. 
For a set $\mc X$, $\partial\mc X$ denotes its boundary.
Capital bold letters represent matrices with appropriate dimensions. For $\bW$, the notation $\bW_{[j,k]}, \bW_{[i,:]}$ denotes the matrix's entry at row $j$ and column $k$, and its $i$th row, respectively.
For a matrix $\bPf$, $\bPf\succ 0\ (\bPf\succeq 0)$ implies that it is positive (semi)definite. For a vector $\bx$, $\bPf\succ 0$, $\|\bx\|_\bPf = \sqrt{\bx^\top\bPf\bx}$.
The vector $\nabla V(\bx)$ denotes the gradient of the function $V(\bx).$
Let $\mc E(\bPf,\varepsilon)$ denote the ellipsoid $\{\bx:\|\bx\|_\bPf \leq \varepsilon\} $. For a set $\mc B$, Vol$(\mc B)$ denotes its volume.
For a neural network, $\bpf^{[\ell]}, \bW^{[\ell]}, \bb ^{[\ell]}$ denote the output vector, weight and bias of its $\ell$th layer, respectively. \thinh{The time variable is denoted as $t.$}

\section{Preliminaries and problem formulation}
\label{sec:preliminaries}
\subsection{Differentially flat systems}
\label{subsec:flat}
In this work, we consider the class of single-input differentially flat systems described as:
\begin{subnumcases}
{
 \label{eq:flatGeneral}
}
\dot \bx = f(\bx) + g(\bx)u, & \label{eq:flatGeneral_dyna}\\
\text{s.t. }u\in\mc U =\{u\in\R: |u|\leq \bar u\}&\label{eq:flatGeneral_constr} 
\end{subnumcases}
where $\bx\in\R^n$ and $u\in\R$ denote the state and input respectively, and the definition of flat systems is given as follows.

\begin{definition}[Differential flatness \citep{fliess1995flatness}]
    System \eqref{eq:flatGeneral_dyna} is called differentially flat if there exist an output \be y=\phi_y(\bx,u),\ee  such that $\bx$ and $u$ can be algebraically expressed in function of $y$ and a finite number of its derivatives, i.e., 
\be
\bx = \Phi_x\left(y,\dot y, ..., y^{(n-1)}\right),
        u = \Phi_u\left(y,\dot y, ..., y^{(n)}\right).
\label{eq:flatrep}
\ee 
\end{definition}
With the representation \eqref{eq:flatrep}, an important property is that the system \eqref{eq:flatGeneral_dyna} can be exactly linearized by static feedback. Namely, by a diffeomorphic variable change and an input transformation:
\be 
\bx = \phi_x(\bz), \ u=\phi_u(\bz,v), \label{eq:linearizingMaps}
\ee 
the dynamics \eqref{eq:flatGeneral_dyna} become:
\be 
\dot \bz =  \bA \bz +\bB v,
\label{eq:linearizedFlat}
\ee 
where $\bz\triangleq [y, \ \dot y,..., y^{(n-1)}]^\top$ and $\phi_u(\bz,v)$ is found by replacing $\bz$ and $y^{(n)}=v$ into \eqref{eq:flatrep}. The matrices $\bA,\bB$ are from the state-space form of the system $y^{(n)}=v$.

Technically, the chain of integrators described in \eqref{eq:linearizedFlat} is straightforward to stabilize \citep{sontag1990nonlinear}. Yet, under the presence of the constraint \eqref{eq:flatGeneral_constr}, the control becomes difficult, since in the new coordinate $(\bz,v)$, one has a state-input joint constraint:
\be 
(\bz,v)\in\mc V\triangleq\left\{(\bz,v):|\phi_u(\bz,v)|\leq \bar u \right\}.
\label{eq:convolutedConstr}
\ee 

\textit{Objective:} With this setting, our goal is to synthesize a controller to stabilize the trivial system \eqref{eq:linearizedFlat} subject to the convoluted constraint 
\eqref{eq:convolutedConstr} 
towards an equilibrium point $\bz_e = \boldsymbol{0}, v_e  =0$
with a stability guarantee. We also assume that this equilibrium point corresponds to a equilibrium point in the original coordinates, noted as $\bx_e=\boldsymbol{0}, u_e=0.$

For this aim, in our prior work \citep{DoMILP2024}, we showed that the constraint \eqref{eq:convolutedConstr} can be enforced using ReLU-ANN. \thinh{Before proceeding, we assume that the mapping $\phi_u(\bz,v)$ is continuous in a compact domain of interest, ensuring the applicability of ReLU-ANNs as universal function approximators.}

\subsection{ANN-based constraint description}

\begin{figure}[htbp]
    \centering
    \includegraphics[width=0.35\linewidth]{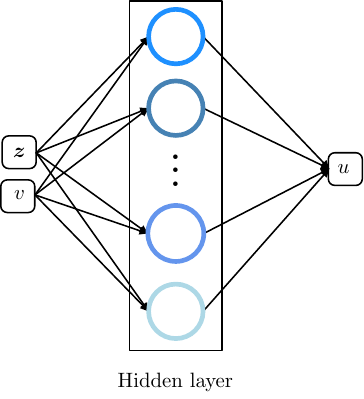}
    \caption{The approximation $\tilde\phi_u(\bz,v)$ of $\phi_u(\bz,v)$ with the structure in \eqref{eq:singleLayerANN}.}
    \label{fig:OneHiddenIllus}
\end{figure}
In details,
let us first consider a structure of a single hidden layer ReLU-ANN mapping from $\bpf^{[0]}\in\R^{n_0}$ to $\bpf^{[2]}\in\R$:
\be 
\bpf^{[1]} = \sigma\left(\bW^{[1]}\bpf^{[0]}+\bb^{[1]} \right), \bpf^{[2]} =\bW^{[2]}\bpf^{[1]}+\bb^{[2]}.\label{eq:singleLayerANN}
\ee 

Then, the constraint \eqref{eq:convolutedConstr} can be imposed with the following proposition from \citep{DoMILP2024}.

\begin{prop}[\thinh{ANN-based constraint inner approximation} \citep{DoMILP2024}]
\label{prop:ANNconstraint}
    Consider a ReLU-ANN approximation of $\phi_u(\bz,v)$ in \eqref{eq:linearizingMaps}, denoted as $\tilde\phi_u(\bz,v)$ structured as in \eqref{eq:singleLayerANN}. Let the approximation error be bounded as:
    \be 
\sup_{(\bz,v)\in\mc Z_w}\left|\phi_u(\bz,v)- \tilde\phi_u(\bz,v)\right|\leq \epsilon,\label{eq:boundedApproximationError}
    \ee 
for some workspace of interest $\mc Z_w$. Then, the constraint \eqref{eq:convolutedConstr} is implied by:
 \be 
(\bz,v)\in\tilde{\mc V}\triangleq\{(\bz,v): |\tilde\phi_u(\bz,v)| \leq \bar u -\epsilon\}. \label{eq:tightenedConstr}
 \ee
Furthermore, the constraint \eqref{eq:tightenedConstr} can be represented by mixed-integer (MI) linear constraints. Briefly, by assigning a binary variable to each neurons, the conditional activation of ReLU function can be modeled as MI linear inequalities.
\end{prop}
\begin{proof}
    \thinh{The inference from \eqref{eq:tightenedConstr} to \eqref{eq:convolutedConstr} can be shown by the triangle inequality. Indeed, one has
    \be 
    \begin{aligned}
        |u |=  |\phi_u(\bz,v)| &= |\phi_u(\bz,v) - \tilde\phi_u(\bz,v) + \tilde\phi_u(\bz,v)| 
         \\
          & \leq |\phi_u(\bz,v) - \tilde\phi_u(\bz,v)| + |\tilde\phi_u(\bz,v)| 
         \overset{\eqref{eq:boundedApproximationError}}{\leq} \epsilon + |\tilde\phi_u(\bz,v)|.
    \end{aligned}
    \ee 
    Hence if $(\bz,v)\in\tilde{\mc V}$ or $\epsilon + |\tilde\phi_u(\bz,v)|\leq \bar u $, one has:
 $
     |\phi_u(\bz,v)| \leq \bar u,
    $    completing the proof. For the derivation of the MI representation, we refer to \citep{DoMILP2024} for more details.}
\end{proof}

In addition, with the piecewise-affine (PWA) representation of ReLU-ANN, we also showed that the set $\tilde{\mc V}$ in \eqref{eq:tightenedConstr} can be described as a union of polytopes
\citep{thinhPWAANN}:
\be 
\tilde{\mc V} =\hspace{-0.15cm}\bigcup_{\balpha\in\{\pm 1\}^{n_1}} \{\bzeta :|\boldsymbol\Psi(\boldsymbol\alpha,\Pi)\bzeta+\boldsymbol{\psi}(\boldsymbol{\alpha},\Pi)|\leq\bar u -\epsilon\}, \label{eq:unionOriginal}
\ee 
where we use $\Pi$ to denote the network parameters (all the weights and biases) and $\balpha = [\alpha_1,...,\alpha_{n_1}]^\top$ is an integer vector. The $i$th row of $\boldsymbol\Psi(\boldsymbol\alpha,\Pi),\boldsymbol\psi(\boldsymbol\alpha,\Pi)$ are computed as:
\begin{subnumcases}
    {\label{eq:cellenum}}
   \boldsymbol\Psi_{[j,:]}(\boldsymbol\alpha,\Pi) =\textstyle\sum_{k=1}^{n_1}\frac{1}{2}\boldsymbol{W}^{[2]}_{[j,k]}(\alpha_k+1)\boldsymbol{W}^{[1]}_{[k,:]},&\\
  \boldsymbol\psi_{j}(\boldsymbol\alpha,\Pi) = \textstyle\sum_{k=1}^{n_1}\frac{1}{2}\boldsymbol{W}^{[2]}_{[j,k]}\bb^{[1]}_k(\alpha_k+1) + \bb^{[2]}_j , &
\end{subnumcases}
for $j\in\{1,...,n_2\}$ and $\bzeta=[\bz^\top\ v]^\top$.

For compactness, after computing the set in \eqref{eq:unionOriginal} and eliminating all empty sets, we obtain $\tilde{\mc V} $ as:
\be
\tilde{\mc V}=\bigcup_{i=1}^{N_c} \mc C_i =\bigcup_{i=1}^{N_c}\{\bzeta=\bbms \bz\\ v\ebms :\boldsymbol{\Xi}_i\bzeta \leq \bxi_i\}, \label{eq:ANNconstr_compact}
\ee 
with $\mc C_i$ being the non-empty polytope from the union \eqref{eq:unionOriginal}.



Next, within the set-theoretic framework, we recall the notion of positively invariant (PI) set to characterize the stability of the system.
\begin{definition}[Positive invariance \citep{blanchini2008set}]
\label{def:invariantSet}
Given a feedback $u=\mu(\bx)$, the set $\mc S$ is positively invariant for the system $\dot\bx=f(\bx,\mu(\bx))$ if
$\forall t>0$, $\bx(t)\in \mc S$ where $\bx(t)$ is the solution of the differential equation with the initial state $\bx(0)\in \mc S$.
\end{definition}

With Proposition \ref{prop:ANNconstraint}, the control problem now is how to design a control law respecting the convoluted constraint \eqref{eq:tightenedConstr}. In our prior work \citep{DoMILP2024,thinhPWAANN}, the constraints \eqref{eq:tightenedConstr} and \eqref{eq:ANNconstr_compact} are enforced directly online via MIP, which is computationally demanding. In this work, we intend to shift the computational burden offline. More specifically, by choosing apriori a state feedback control associated with a CLF, an invariant set will be identified offline, certifying the stability property. Subsequently, only the CLF and the invariant set are retained for online control synthesis, rather than the complicated constraint \eqref{eq:tightenedConstr}.

\section{Invariance characterization and control synthesis}
\label{sec:Invariancecharacterization}

As discussed above, having the representation \eqref{eq:linearizedFlat} gives us the access to well-developed tools from linear control theory. In what follows, we show that, with a static gain state feedback, one can certify an invariant set compatible with the constraint \eqref{eq:tightenedConstr}.
Throughout the first part, the problem of invariance characterization is defined and
reformulated into an optimization problem from a set-theoretic viewpoint. This characterization then gives rise to the applications to advanced tools such as CLF-based control and MPC in the subsequent part.

\subsection{Invariant set characterization from the flat output space}
\label{subsec:InvSr}

First, for the system \eqref{eq:linearizedFlat}, consider a state feedback:
\be 
v = \boldsymbol{K} \bz, \label{eq:stateFeedbackFlat}
\ee
which renders the Lyapunov function
\be
V_z(\bz) = \bz^\top \bPf\bz,
\ee 
decreasing along the system's trajectory, i.e.,
\be 
\nabla V_z(\bz)^\top (\bA +\bB\bK)\bz \leq -\kappa V_z(\bz), \label{eq:CLF_derivative}
\ee 
for some scalar $\kappa>0$.
Depending on the control framework, the matrices $\bPf,\bK$ can be found by different means (e.g., linear quadratic regulator tunings or linear matrix inequalities). Then, an invariant set associated with a control law for \eqref{eq:flatGeneral} can be formulated with the following proposition.
\begin{prop}\label{prop:invariantFlat}
    The ellipsoid $\mc E(\bPf,\varepsilon)$ satisfying:
\be \
\mc E(\bPf,\varepsilon) \subseteq \mc Z_K\triangleq\{\bz: (\bz,\bK\bz)\in \tilde{\mc V} \text{ as in \eqref{eq:ANNconstr_compact}}\},\label{eq:inclusionFlat}
\ee 
is PI for the system \eqref{eq:linearizedFlat}. Furthermore, the set:
\be 
\mc X_{K} = \{\bx : V_x(\bx) \leq \varepsilon\}, \label{eq:invariantStateInputspace}
\ee 
is also  under the control: \be u = \mu_K(\bx) \triangleq \phi_u(\phi^{-1}_x(\bx),\bK\phi^{-1}_x(\bx)) , \label{eq:admissbleControlXU}\ee
with the CLF $V_x(\bx) $ defined as:
\be 
V_x(\bx) = \phi^{-1}_x(\bx)^\top \bPf\phi^{-1}_x(\bx), \label{eq:CLF_original}
\ee 
and $\phi^{-1}_x(\bx)$ denoting the inverse mapping of $\bx = \phi_x(\bz)$.
\end{prop}
\begin{proof}
    The proof is twofold.
    \thinh{First, the control law \eqref{eq:stateFeedbackFlat} is admissible inside $\mc E(\bPf,\varepsilon)$ due to the inclusion \eqref{eq:inclusionFlat}. More specifically, 
    if $\bz \in \mc E(\bPf,\varepsilon)$, one has $(\bz,v)\in\tilde{ \mc V}$ as $v=\bK \bz$ from \eqref{eq:stateFeedbackFlat}. Then, by Proposition \ref{prop:ANNconstraint}, it follows that $ |u |=| \phi_u(\bz,v)| \leq \bar u.$ Hence the control law is constraint admissible.
}
Second, the ellipsoid $\mc E(\bPf,\varepsilon)$ is a level set of the convex function $V_z(\bz)$ which is rendered decreasing by the choice of $\bK$ satisfying \eqref{eq:CLF_derivative}. 
\thinh{Namely, inside $\mc E(\bPf,\varepsilon)$, the value $V(\bz(t))$ decreases along the system's trajectory.
With the invertibility of $\phi_x(\bx)$ stemming from the theory of differential flatness, 
one has $ 
V(\bz(t)) = V_x(\bx(t)),
$ 
and $ \bz \in \mc E(\bPf,\varepsilon)$ if and only if $\bx \in\mc X_K$
 in \eqref{eq:invariantStateInputspace}. Therefore, the value of $V_x(\bx(t))$ is also decreasing along the trajectory $\bx(t) = \phi_x (\bz(t))$ controlled by \eqref{eq:stateFeedbackFlat}.
} In other words, the set $\mc X_K$
is also PI associated with the CLF in \eqref{eq:CLF_original}.
\end{proof}

\begin{rem}
It can be understood that $\mc Z_K$ in \eqref{eq:inclusionFlat} is a set of the new state $\bz$ inside which no constraint violation will occurs. With this interpretation, the inclusion condition \eqref{eq:inclusionFlat} implies that $\mc E(\bPf,\varepsilon)$ is an invariant set in which $v=\bK \bz$ is admissible with respect to the constraint \eqref{eq:convolutedConstr} and eventually, \eqref{eq:flatGeneral_constr}.
\end{rem}

\begin{rem}
\label{rem:uniOn}
It should be noted that $\mc Z_K$ in \eqref{eq:inclusionFlat} is also a union of polytopes, because by \eqref{eq:stateFeedbackFlat} and \eqref{eq:ANNconstr_compact}, one has:
\be
\mc Z_K =\bigcup_i^{N_c} \left\{\bz:\boldsymbol{\Xi}_i\bbm \boldsymbol{I} \\ \boldsymbol{K}\ebm \bz\leq \bxi_i\right\}. \label{eq:ZK_union}
\ee 
This representation will prove useful for the subsequent stability certification.
\end{rem}

With Proposition \ref{prop:invariantFlat}, choosing $\bK, \bPf$ satisfying \eqref{eq:stateFeedbackFlat}--\eqref{eq:CLF_derivative} is rather simple. For instance, one standard choice is to parameterize \citep{blanchini2008set}: 
\be \bPf = \boldsymbol{\Upsilon}^{-1}, \bK = -\bB^{\top}\bPf, \label{eq:gainPara}\ee
where the symmetric matrix $\boldsymbol{\Upsilon}$ is a solution of the linear matrix inequalities:
\be 
\boldsymbol{\Upsilon} \bA^\top + \bA\boldsymbol{\Upsilon} - 2\bB\bB^\top \preceq - 2\kappa\boldsymbol{\Upsilon}, \boldsymbol{\Upsilon} \succ 0. \label{eq:LMIGainPara}
\ee 

Meanwhile, the primary challenge persists in the search for the level set $\mc E(\bPf, \varepsilon)$ fulfilling \eqref{eq:inclusionFlat}. Formally, this problem can be translated to solving the optimization problem:
\be 
\begin{cases}
    \bz ^*= \underset{\bz\in\partial\mc Z_K}{\arg\min}\ \bz ^\top\bPf \bz, & \\
    \varepsilon=V(\bz^*).
\end{cases}
\label{eq:opti_maxEll}
\ee 
Geometrically, the rationale behind \eqref{eq:opti_maxEll} is to find the weighted norm distance
from the origin to the boundary $\partial\mc Z_K$ of $\mc Z_K$ (see Fig. \ref{fig:biggestEll_illust} for illustration).

\begin{figure}[htbp]
    \centering
    \includegraphics[width=0.45\linewidth]{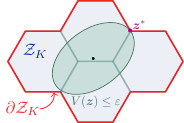}
    \caption{Distance from the origin to the set $\partial\mc Z_K$.}
    \label{fig:biggestEll_illust}
\end{figure}

With the fact that $\mc Z_K$ is a union of $n$-dimensional polytopes as shown in Remark \ref{rem:uniOn}, its boundary is evidently also a union of $(n-1)$-dimensional polytopes. Computing these polytopes, in general for arbitrary dimension still suffers from the curse of dimensionality \citep{aronov1997union,jordan2019provable}. \thinh{While a detailed complexity analysis is beyond the present scope, related discussions on this scalability issue can be found in \citep{karg2020efficient,montufar2014number}.}
For our low-dimensional examples, one brute-force solution is to take advantage of the half-space representation \eqref{eq:ZK_union} and simply collect all the boundary facets. 
For the sake of notation, we denote the boundary in the form:
\be 
\partial\mc Z_K=\bigcup_{j=1}^{N_b}\left \{\bz :\boldsymbol{\Theta}_j\bz \leq \boldsymbol\theta_j\right\}. \label{eq:Boundary_union}
\ee 
Consequently, the problem \eqref{eq:opti_maxEll} can be formulated as:
\be 
\varepsilon=\min_{1\leq j\leq N_b}\underset{\boldsymbol\Theta_j\bz\leq \boldsymbol\theta_j}{\min}\ \bz^\top \bPf \bz, \label{eq:manyQPs}
\ee 
which is simply a search among the solutions of $N_b$ quadratic programs (QPs). Alternatively, the value $\varepsilon$ in \eqref{eq:manyQPs} can also be computed by solving a mixed-integer QP \citep{prodan2015mixed}:
\be 
\begin{aligned}
  &\quad\quad\quad\quad\quad\quad\quad\quad\varepsilon=  {\min}\ \bz^\top \bPf \bz  \\
  & \begin{cases}
      \boldsymbol{\Theta}_j\bz \leq 
      \boldsymbol{\theta}_j + M\delta_j,  
      \delta_j \in \{1,0\}, j\in\{1,...,N_b\},& \\ 
      \sum_{j=1}^{N_b}= N_b-1,& 
  \end{cases}
\end{aligned} \label{eq:InvarianceMIQP}
\ee 
where the scalar $M>0$ is the ``big-M" chosen sufficiently large to model the union \eqref{eq:Boundary_union}.

In summary, with a ReLU-ANN satisfying Proposition \ref{prop:ANNconstraint} with the PWA description \eqref{eq:ANNconstr_compact}, one can detect an invariant set $\mc E(\bPf,\varepsilon)$ by following Procedure \ref{proc:PICharac}.





\begin{procedure}
\label{proc:PICharac}
  \caption{Compute a PI from the flat output space()}
 Choose a feedback gain $\bK$ associated with a matrix $\bPf$ satisfying \eqref{eq:stateFeedbackFlat}--\eqref{eq:CLF_derivative}\; 
     Compute the set $\mc Z_K$ as in \eqref{eq:ZK_union} and collect its boundary $\partial\mc Z_K$ in the form of \eqref{eq:Boundary_union}\;
     Find the invariant level set $\mc E(\bPf,\varepsilon)$ as in \eqref{eq:inclusionFlat} via \eqref{eq:InvarianceMIQP}\;

\end{procedure}


Next, to emphasize the applicability of the setting, we will show the CLF and its associated invariant set can be employed in a control design process. Particularly for the MPC design, we present a procedure to choose the feedback gain and the quadratic CLF.

\subsection{Applications to control design}
\label{subsec:ControlDesign}

It has been shown that with the characterized PI set, the controller $\bK\bz$ in \eqref{eq:stateFeedbackFlat} or $\mu_K(\bx)$ in \eqref{eq:admissbleControlXU} is an constraint-admissible controller certifying Lyapunov stability inside the set. Having this interpretation, several control design framework can be applied to enhance the performance with respect the such nominal laws. In what follows, we enumerate some control techniques which can take advantage of the stability certificate to regulate the nonlinear system.

\subsubsection{CLF-based control}
\phantom{abc}

Given the CLF $V_x(\bx)$ in 
\eqref{eq:CLF_original}, various established methodologies become directly applicable. For instance, Sontag's ``universal" explicit formula \citep{sontag1989universal} or the min-norm control \citep{freeman1996inverse}. In this work, we present the minimally invasive control as a representative case:

\begin{subequations}
{\label{eq:CLFgeneral}}
\begin{align}
       &  u = \underset{|u|\leq\bar u}{\arg \min} \|u - u_d(\bx)\|_2^2,  \\
        &\text{s.t. }\nabla V_x(\bx)^\top (f(\bx) + g(\bx)u)\leq -\kappa V_x(\bx) , \label{eq:CLFgeneral_clf}
\end{align}
\end{subequations}
where $u_d(\bx)$ is some desired controller for $u$ to track. Meanwhile, by the chain rule, the gradient $\nabla V_x(\bx)$ can be computed as:
\be 
\nabla V_x(\bx) = \mc D \phi_x(\bx)^\top(\bPf+\bPf^\top)\phi_x(\bx),
\ee 
with $\mc D \phi_x(\bx)$ denoting the Jacobian of the vector value function $\phi_x(\bx)$.
The implicit control \eqref{eq:CLFgeneral} presents a QP when given a state feedback $\bx$ with stability guarantee from the condition \eqref{eq:CLFgeneral_clf}.
It is important to note that, 
the simplicity of \eqref{eq:CLFgeneral} is enabled by the identification of function $V_x(\bx)$ with an admissible control $u=\mu_K(\bx)$ in \eqref{eq:admissbleControlXU}, which is generally non-trivial to determine for nonlinear system \eqref{eq:flatGeneral}.

A shortcoming of the control \eqref{eq:CLFgeneral} is that the tracking is only guaranteed inside $\mc X_K$ as in \eqref{eq:invariantStateInputspace}. One typical remedy for this issue is to employ predictive control where, from an initial state $\bx(0)$ possible outside $\mc X_K$, we can still guarantee the convergence if the predicted trajectory enters the set after a certain steps. This condition is commonly referred to as the terminal constraint with a proper choice of tunings.
Thus, in the next part, we show how the invariant set can be integrated with the standard axioms from MPC design.

\subsubsection{Flatness-based MPC with stability guarantee}
\phantom{abc}

 In details, the synthesis of the MPC 
 can be given in the following proposition.
\begin{prop}
    \label{prop:MPCtunings}
Given the dynamics \eqref{eq:linearizedFlat} subject to the constraint $(\bz,v)\in\mc V$  as in \eqref{eq:convolutedConstr} with a tightened set $\tilde{\mc V}$ as in \eqref{eq:ANNconstr_compact}, and some user-defined weightings $\bQf,\bRf\succ 0$, consider the following procedure:
\begin{itemize}
    \item Choose a gain $\bK^*$, rendering $(\bA+\bB\bK^*)$ Hurwitz;
    \item Choose $\boldsymbol{M}$ such that:
    \be
    \boldsymbol{M} \succ \bQf + \bK^{*\top}\bRf\bK^*; \label{eq:MatrixM_MPC}
    \ee
    \item Find $\bPf^*$ from the Lyapunov equation:
    \be 
(\bA+\bB\bK^*)^{\top}\bPf^* + \bPf^* (\bA+\bB\bK^*)+\boldsymbol{M} = \boldsymbol 0; \label{eq:GainP_MPC}
    \ee 
    \item Define the set \be \mc Z_p\triangleq\{\bx : \|\bz\|_{\bPf^*}^2\leq \varepsilon^*\} \label{eq:TerminalSet}\ee  with $\varepsilon^*$ is the solution of \eqref{eq:InvarianceMIQP} with $\bK=\bK^*$ and $\bPf=\bPf^*$ .
\end{itemize}
Then we have,
\begin{itemize}
    \item[a)] the control $v=\bK^*\bz$ makes the set $\mc Z_p$ become PI and the constraint $(\bz,\bK^*\bz) \in\mc V$ in \eqref{eq:convolutedConstr} is satisfied inside $\mc Z_p$; 
    \item[b)] $ \dfrac{\mathrm{d}F(\bz,\bK^*\bz)}{\mathrm{d}t}+\ell(\bz,\bK^*\bz )\leq 0,\ \forall\bz \in \mc Z_p,$ with
    \be
F(\bz)=\|\bz\|^2_{\bPf^*}, \ \ell(\bz,v) = \|\bz\|^2_{\bQf}+\|v\|^2_{\bRf}.
    \ee 
\end{itemize}

\end{prop}

\begin{proof}
    The first property follows directly from the solution of \eqref{eq:InvarianceMIQP} and Proposition \ref{prop:invariantFlat}. Meanwhile, the second property can be derived directly from \eqref{eq:GainP_MPC} with the choice of matrix $ \boldsymbol{M}$ satisfying \eqref{eq:MatrixM_MPC}.
\end{proof}
\begin{coll}
Consider the receding horizon control at time step $t$:
    \begin{subequations}
    {\label{eq:MPC_flatSpace}}
        \begin{align}
            & { \min} \int_0^{T_p} \|\bz(\tau)\|_\bQf+\|v(\tau)\|_\bRf \mathrm d\tau+\|\bz(T_p)\|_{\bPf^*}  \\
            \text{s.t. } &\dot\bz=\bA\bz+\bB v, \bz(0)=\bz(t),\\
           &(\bz(\tau),v(\tau))\in\mc V \text{ as in \eqref{eq:convolutedConstr}}, \tau\in [0,\ T_p],
           \label{eq:convolutedConstr_predictive}\\ 
           & \bz(T_p) \in \mc Z_{p}\text{ as in \eqref{eq:TerminalSet}},
        \end{align}
    \end{subequations}
    with the prediction horizon $T_p$, and the first control action found $v^*(0)$ will be applied to the system by $u=\phi_u(\bz(t),v^*(0))$. Then, with the synthesis from Proposition \ref{prop:MPCtunings}, the policy \eqref{eq:MPC_flatSpace} is recursively feasible and the closed-loop is asymptotically stable.
\end{coll}

\begin{proof}
    See Section 3.6 in \citep{mayne2000constrained}.
\end{proof}

    


It is clear that, for a large initial error, one can increase the prediction horizon $T_p$ until a feasible solution found. However, this evidently comes with a rise in the computational cost. In what follows, we present a particular setting where one can employ the so-called Explicit Reference Governor approach when a significant change in reference is introduced.










\subsubsection{Explicit 
Reference Governor} \label{subsubsec:RG}
\phantom{abc}

For this particular context, we no longer assume the set-point to track as the origin. Rather, suppose that the equilibria of \eqref{eq:flatGeneral}  can be parametrized by a reference signal $r\in\R$ as:
\be 
\bx_e =\bx_e(r), u_e=0.
\ee 
We further assume that $\phi_u(\bz,v)$ in \eqref{eq:linearizingMaps} is independent of $r$ (i.e., $\partial\phi_u(\phi_x^{-1}(\bx_e(r),v)/\partial r=0$). In this manner, the set:
\be 
\bar{\mc X}_{K}(r) = \{\bx: \bar V_x(\bx,r)\leq \varepsilon\}, \label{eq:safeLevelSet}
\ee 
is PI for a fixed reference $r$ under the control:
\be
\bar\mu_{K}(\bx,r) = \phi_u(\phi^{-1}_x(\bx),\bK(\phi^{-1}_x(\bx)-\phi^{-1}_x(\bx_e(r))), \label{eq:PrimaryControlERG}
\ee 
with $\bar V_x(\bx,r)$ adapted from \eqref{eq:CLF_original} as:
\be 
\bar V_x(\bx,r) = \|\phi^{-1}_x(\bx)-\phi^{-1}_x(\bx_e(r))\|_{ \bPf},
\label{eq:CLF_original_ref}
\ee 
and $\varepsilon$ is found from Proposition \ref{prop:invariantFlat} and \eqref{eq:InvarianceMIQP}.

Then, following the framework of Explicit Reference Governor (ERG), given a desired reference $r_d$, one can compute a feasible reference $r_f(t)$ such that $\bar\mu_{K}(\bx,r_f(t))\in\mc U$ as in \eqref{eq:flatGeneral_constr} and $r_f(t) \to r_d$ when $t\to \infty$. More specifically, $r_f(t)$ can be computed by integrating:
\be 
\dot r_f = \rho(r_d,r_f)\Delta(\bx,r_f), \label{eq:ERG_integration}
\ee 
where $\rho(r_d,r_f)$,  $\Delta(\bx,r_f)$ are called the attraction field and the dynamic safety margin, respectively. Their value are given as \citep{nicotra2018explicit}:
\be 
\begin{cases}
    \rho(r_d,r_f)&= (r_d-r_f)/\max(|r_d-r_f|,\eta), \\
    \Delta(\bx,r_f)& = \lambda(\varepsilon-\bar V_x(\bx,r_f) ),
\end{cases} \label{eq:ERGingredients}
\ee 
with $\lambda>0$ is a tuning parameter and $\eta>0$ is a small scalar to avoid singularity. The scalar $ \varepsilon$ characterizes the level set $\bar{\mc X}_K(r)$ as in \eqref{eq:safeLevelSet} inside which the constraint is respected.
Essentially, the online program includes, at each time step, computing $r_f(t)$ by \eqref{eq:ERG_integration} and applying the control $u = \bar \mu_K(\bx,r_f)$ as in \eqref{eq:PrimaryControlERG} to the system.

To this point, it can be seen that with the system's linear representation \eqref{eq:linearizedFlat} in the flat output space and the approximated convoluted constraint \eqref{eq:ANNconstr_compact}, one can compute a CLF $V_z(\bz)$ with the tunings from linear control theory and detect the corresponding invariant set by searching for its level set inscribed in the convoluted constraint.
With such invariance characterization, by properly choosing the control gains or setting up the online control, we have shown that stabilizing policies can be synthesized from different frameworks. In the next part, simulation results will be provided to demonstrate the validity of the proposed setup.


\section{Illustrative examples}
\label{sec:IllusExam}

We validate the proposed setting by two examples.

\subsection{Aircraft longitudinal dynamics stabilization}
\label{subsec:aircr}

For demonstration, let us examine the longitudinal dynamics of a civil aircraft model \citep{nicotra2018explicit}:
\be 
\begin{aligned}
    \dot x_1 & = x_2 \\
    \dot x_2 & = J^{-1}(-d_1L(x_1)+ud_2)\cos x_1
\end{aligned}
\label{eq:longitudinalDyna}
\ee 
with $\bx =[x_1 \ x_2]^\top$ denoting the state vector and $x_1$ being the angle of attack, $L(x_1)$ is the lift generated modeled as $L(x_1) = l_0+l_1x_1+l_3x_1^3$. The input $u$ is the elevator force. The constraint of the system is
$
|u| \leq \bar u.
$
The constant parameters are given as  $l_0=2.5\times 10^5,l_1 = 8.6\times 10^6, l_3 = 4.35\times 10^7$, $J=4.5\times10^6$Nm$^2$, $\bar u = 5\times 10^5$ N, $d_1 = 4 $m and $d_2  =42$m.
The goal is to stabilize the system at $\bx_e = \boldsymbol{0}, u_e = d_1l_0/d_2.$ The mappings as in \eqref{eq:linearizingMaps} for \eqref{eq:longitudinalDyna} are $\phi_x(\bz) = [z_1\ z_2]^\top, $ $\phi_u(\bz,v) = {d_2^{-1}(vJ\cos^{-1}z_1+d_1L(z_1)) }$. As a result, in the flat output space, the linearized dynamics become $\ddot z_1 = v$, or the form of \eqref{eq:linearizedFlat} with $\bA = \bbms 0 &1 \\ 0& 0  \ebms, \bB = [0 \ 1]^\top$. 
\begin{figure}[htbp]
    \centering
    \includegraphics[width=0.8\linewidth]{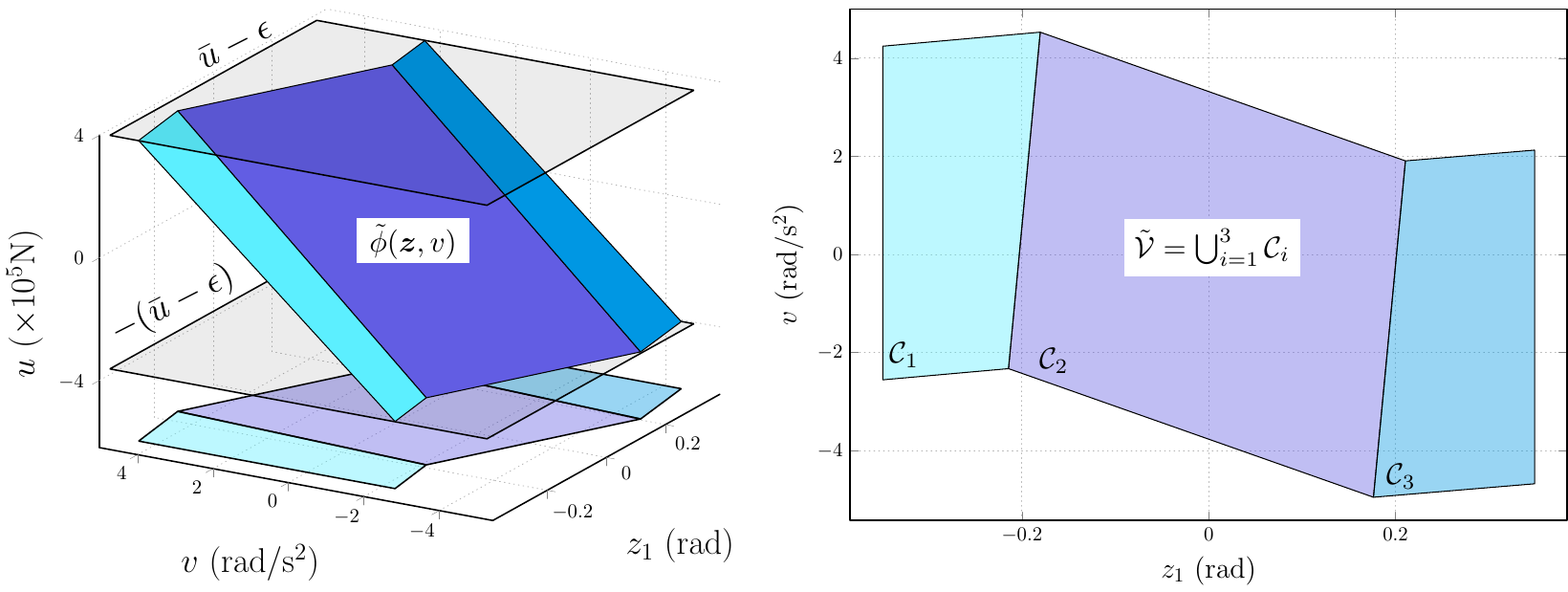}
    \caption{The constrained approximated mapping $\tilde{\phi}_u(\bz,v)$ (left) and the corresponding constraint set $\tilde{\mc V}$ induced from the neural network as in \eqref{eq:cellenum} (right).}
    \label{fig:IllusAircraftCellEnum}
\end{figure}

To approximate $\phi_u(\bz,v)$, we employed an ANN $\tilde\phi_u(\bz,v)$ structured as in \eqref{eq:singleLayerANN} with 5 neurons in the hidden layer. 
The training data is obtained from a uniform grid inside  $\{(\bz,v):|z_1|\leq 0.3491,|v|\leq 5\}.$
The approximation error $\epsilon$ as in \eqref{eq:boundedApproximationError} is estimated as $\epsilon=0.1897$. \thinh{In the paper, such an error bound is verified by uniformly gridding the joint space $(\bz,v)$ and evaluating the maximum deviation between the ANN $\tilde\phi_u(\bz,v)$ and the nonlinear mapping $\phi_u(\bz,v)$. A discussion on the formal guarantee of the error bound with respect to the grid size can be found in Proposition 4 of \citep{thinhPWAANN}}.
With the explicit representation \eqref{eq:unionOriginal}, the constraint set $\tilde{\mc V} $ in \eqref{eq:unionOriginal} is represented as the union of 3 polytopes (see Fig. \ref{fig:IllusAircraftCellEnum}).
Geometrically, the set $\tilde{\mc V}$ is formed by truncating the approximated network $\tilde\phi_u(\bz,v)$ by the absolute input bounds $\bar u-\epsilon$, creating the non-convex feasible domain as depicted on Fig. \ref{fig:IllusAircraftCellEnum}.

\begin{figure}[htbp]
    \centering
    \includegraphics[width=0.745\linewidth]{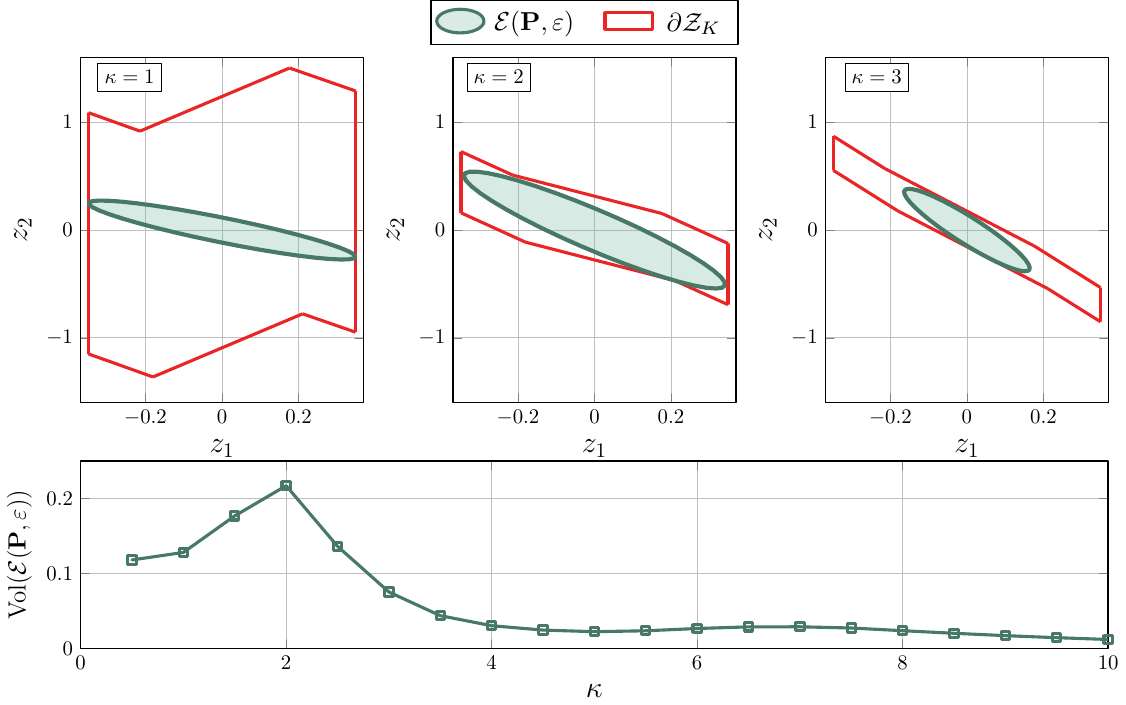}
    \caption{PI sets and their volumes (noted with Vol($\cdot$)) with different convergence rates $\kappa$ as in \eqref{eq:CLF_derivative}.}
    \label{fig:AirCraftMultigains}
\end{figure}

With this constraint set, we follow Procedure \ref{proc:PICharac} to find the ellipsoidal PI where the matrix $\bPf,\bK$ are chosen via \eqref{eq:gainPara}--\eqref{eq:LMIGainPara} with different convergence rate $\kappa$.
The numerical results are depicted in Fig. \ref{fig:AirCraftMultigains}.

\subsection{Horizontal 1D quadrotor}
\label{subsec:quad}

To present the application of the invariance characterization to the control synthesis procedures in Section \ref{subsec:ControlDesign}, we consider the
1-D quadcopter model (Quad1D) \citep{greeff2020exploiting}:
\be  
\begin{aligned}
    \dot x_1 & = x_2, \\
     \dot x_2 &=\Gamma \sin x_3 - \gamma  x_2, \\
      \dot x_3&=\tau^{-1}(u-x_3),
\end{aligned}
\label{eq:quad1D_nonlinear}
\ee 
where $x_1,x_3$ denotes the displacement on the horizontal axis (m) and the pitch angle (rad), respectively (see Fig. \ref{fig:Drone1DIllus}). The control $u$ is the commanded pitch angle subject to the constraint $|u|\leq \bar u= 0.1745$ (rad). The constant parameters are given as $\Gamma=10,\gamma=0.3,\tau=0.2.$ The goal is to stabilize the system at the origin $\bx_e = \boldsymbol{0}, u_e = 0.$

\begin{figure}[h]
    \centering
    \includegraphics[width=0.35\linewidth]{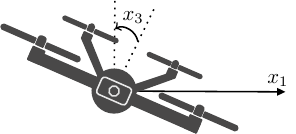}
    \caption{Horizontal quadcopter model.}
    \label{fig:Drone1DIllus}
\end{figure}

For this model, the flatness-based linearizing mapping of the form \eqref{eq:linearizingMaps} is given as:
\begin{subequations}\label{eq:linearizingMapDrone_full}
\begin{align}
\phi_x(\bz)&= \bbm z_1\\ z_2 \\ \thinh{\arcsin(\nu(\bz))}\ebm,\\
    \phi_u(\bz,v)&=\dfrac{\tau(v+\gamma z_3)}{\Gamma\sqrt{1-\nu^2(\bz)} 
}+\thinh{\arcsin(\nu(\bz))}, \label{eq:linearizingMapDrone1D}
\end{align}
\end{subequations} 
with $\nu(\bz)=\Gamma^{-1}(z_3+\gamma z_2).$

With the mapping $   \phi_u(\bz,v)$ in \eqref{eq:linearizingMapDrone1D}, we employ a network of the form \eqref{eq:singleLayerANN} with 5 neurons in the hidden layer. The approximation error is estimated as $\epsilon=0.0026$. This approximation in turns, results in the constraint set $\tilde{\mc V}$ in \eqref{eq:unionOriginal} as a union of 5 polytopes. 

\begin{figure}[htp]
    \centering
    \includegraphics[width=0.98\linewidth]{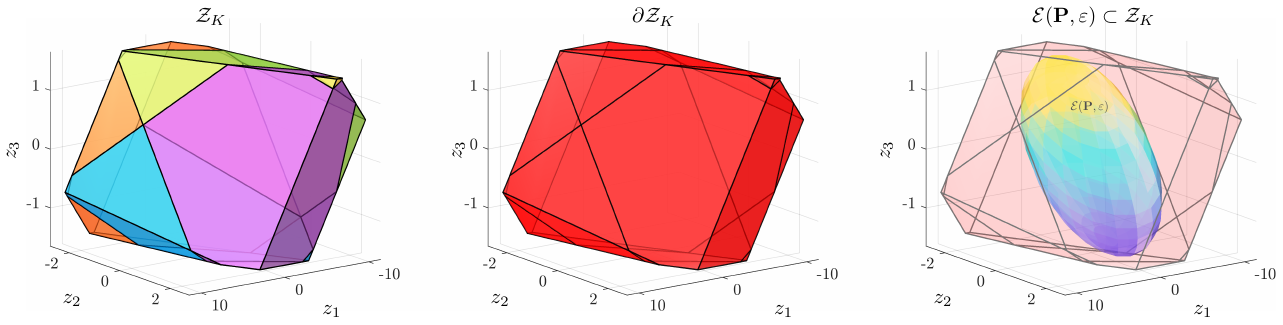}
    \caption{The admissible union of polytopes $ \mc Z_K$ as in \eqref{eq:inclusionFlat} (left), its boundary set (middle) and an ellipsoidal PI set $\mc E(\bPf,\varepsilon)\subset\mc Z_K$ (right).}
    \label{fig:Quad1DBall}
\end{figure}

To demonstrate the effectiveness of the proposed setting, let us consider the following controllers:
\begin{itemize}
    \item \textit{CLF-based control} (CLF): For this setting, we employ the controller \eqref{eq:CLFgeneral} with $u_d(\bx)=\phi_u(\bz,\bK_{lqr}\bz) $, with: $$\bK_{lqr} = -[ 3.2,\    5.5, \    4.0] .$$ 
    The CLF $V_x(\bx)$ is of the form \eqref{eq:CLF_original} with $\bPf$ chosen by solving \eqref{eq:LMIGainPara} with the exponential convergence rate $\kappa = 0.01$. With the nominal control $\mu_K(\bx) =\phi(\bz,\bK\bz) $, $\bK=-\bB^\top\bPf$ and Procedure 
    \ref{proc:PICharac}, the corresponding PI set $\mc E(\bPf,\varepsilon)$ as in \eqref{eq:inclusionFlat} is characterized by:
    \be
    \bPf = \bbm
        0.214   & 0.307   &  0.394 \\
    0.307  &  0.728  &  0.797 \\
    0.394  &  0.797  &  1.290\\
    \ebm, \varepsilon =      0.4329.
    \ee 
    Fig. \ref{fig:Quad1DBall} and \ref{fig:InvarianceTest} depict the procedure for finding the set and verifying its positive invariance.

\item \textit{Explicit Reference governor} (ERG): To show that the PI computed set can also be used with the ERG to improve feasibility, we employ the scheme proposed in Section \ref{subsubsec:RG}. The tuning parameters are chosen as $\lambda=20,\eta=0.2.$ 
    
    \item \textit{Flatness-based MPC} (FMPC): We follow Proposition \ref{prop:MPCtunings} to define the implicit controller as in \eqref{eq:MPC_flatSpace}. The numerical parameters are given as:
    \be
    \begin{aligned}
        \bQf &= \text{diag}(10,10,10),
\bRf=10 ,
\bK^* = [ -1.00,\   -2.41,\   -2.41
]^\top, \\
 \bPf^* &= \bbm
      324.71  & -10.00 & -128.64 \\ 
  -10.00 &104.50  & -34.14 \\
 -128.64 & -34.14&  101.57 \\
    \ebm,
    \varepsilon^* =14.38.
    \end{aligned}
    \ee 
    
\item \textit{Quasi-infinite horizon MPC} (qMPC): To provide a contextual understanding of the computational cost, we employ the qMPC setting proposed in \citep{chen1998quasi} with the same penalty matrices for the state $\bQf$ and the input $\bRf$ as in FMPC. The local controller is derived from the Riccati equation for the Taylor approximation model of the dynamics at the origin.
    
\end{itemize}
\begin{figure}[htbp]
    \centering
    \includegraphics[width=0.90\linewidth]{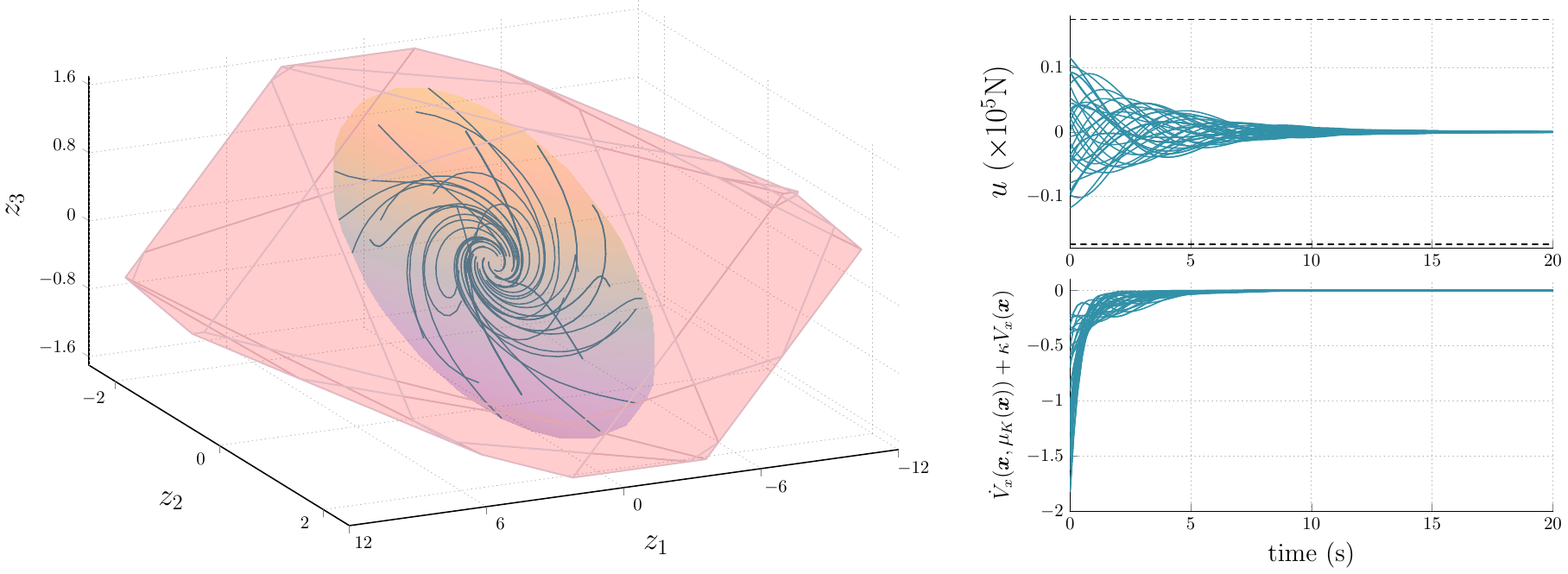}
    \caption{The set $\mc E(\bPf,\varepsilon)$ being PI with the nominal control $\mu_K(\bx)$.}
    \label{fig:InvarianceTest}
\end{figure}

\begin{table}[ht]
  \centering
  \caption{Online computation time (CT) for the controllers with the two scenarios (ms)}
    \begin{tabular}{|c|c|c|c|c|c|}
    \hline
      
    \multicolumn{2}{|c|}{}    &   CLF   &   ERG    & FMPC & qMPC \\
    \hline
    \centering\multirow{2}[0]{*}{ \rotatebox{90}{Case 1}} \rule{0pt}{12pt} &  max CT    &   8.173   &   2.038   &  30.218 &20.359\\
\cline{2-6}   \rule{0pt}{13pt}      &    mean CT  &   4.787   &   0.005   & 8.761 & 7.154\\
    \hline
    \multirow{2}[0]{*}{\centering \rotatebox{90}{Case 2}} &    max CT \rule{0pt}{11pt} &    $\times$  &  2.359    & 54.842&46.770\\
\cline{2-6}   \rule{0pt}{13pt}      &   mean CT   &   $\times$   &   0.005   &   15.236& 12.643\\
    \hline
    \end{tabular}%
  \label{tab:compCost}%
\end{table}%
\begin{figure}[hbt]
    \centering
    \includegraphics[width=0.675\linewidth]{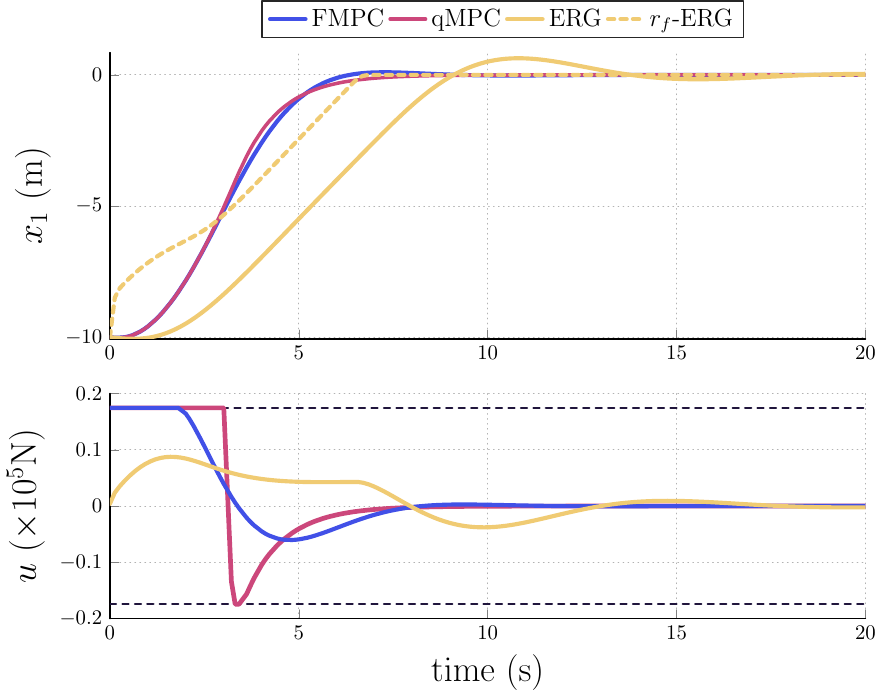}
    \caption{Application of the invariance characterization to different control schemes ($r_f$-RG denotes the filtered reference computed via \eqref{eq:ERG_integration}).}
    \label{fig:CompareControllers}
\end{figure}

To analyze the controllers' applicability, we implement the stabilization from two initial state $\bx(0) = [-2,\ -0.5,\ 0.085]^\top$ (Case 1) and $\bx(0) = [-10,\ -0.15,\ 0]^\top$ (Case 2).
The system's trajectories and the computational cost of the controllers are provided in Fig. \ref{fig:CompareControllers} and Table \ref{tab:compCost}, respectively.

\subsection{Discussion}
\label{subsec:discuss}
In general, via the two examples, ellipsoidal PI sets can be found inside the feasible domain described by the ReLU-ANN in the linearizing or the flat output space (see Fig. \ref{fig:IllusAircraftCellEnum}). The stability of the system inside such domains is correspondingly certified with a CLF of which the finding is simplified thanks to the linear representation \eqref{eq:linearizedFlat} in the new space. 

For the nonlinear dynamics \eqref{eq:longitudinalDyna}, as depicted in Fig. \ref{fig:AirCraftMultigains}, several certified invariant sets can be found with different convergence rates via the parameterization \eqref{eq:LMIGainPara}. This connection to the computation of quadratic CLF of linear systems suggests that one can modify the domain of attraction, suggesting interesting directions such as expanding the stability region towards some desired initial states.

Importantly, inside the computed PI the input constraint will not be violated with the nominal control law $\mu_K(\bx)$ as in \eqref{eq:admissbleControlXU}. More specifically, the characterization indicates that $\mu_K(\bx)$ is a feasible controller which results in the the exponential convergence: $\dot V_x(\bx,\mu_K(\bx))  + \kappa V_x(\bx)\leq 0$ as given in Fig.
\ref{fig:InvarianceTest} without breaking the constraints.

The demonstration on the quadcopter model \eqref{eq:quad1D_nonlinear} also suggests that various control frameworks can take advantage of this characterization to ensure stability with different levels of complexity and performance.
In detail, the computational cost in Table \ref{tab:compCost} for initial point in Case 1 shows that, ERG requires the least computational power since it employs only standard calculation and numerical integration. Next, the complexity of CLF is slightly elevated but still relatively simple. This is because the implemented online optimization \eqref{eq:CLFgeneral} presents a simple QP with three linear inequalities. Finally, although having stability guaranteed, the FMPC \eqref{eq:MPC_flatSpace} presents a nonlinear optimization problem due to the constraint \eqref{eq:convolutedConstr_predictive}, consuming even more computational overhead than the baseline control qMPC. 
Additionally, with the further initial position of the drone in Case 2 (10m from the target), it is evident that this exceeds the certified region of the CLF (see Fig. \ref{fig:InvarianceTest}), rendering this control inapplicable. However, for the predictive controls (FMPC and qMPC), one can simply increase the prediction horizon until a feasible solution is found. Consequently, the complexity of the program becomes higher as indicated by their computation times. Interestingly, thanks to the simple nature and a well-characterized invariance, the required online calculation of ERG time does not vary while the constraint is respected even for such a large initial error.
The trade-off for this simplicity, however, is a slower convergence and input under-exploitation.

Essentially, the practical implementation of control strategies is directly influenced by computational resource at hand, necessitating a compromise between computational complexity and control performance.
Note that the possible application of the characterized PI set does not strictly limit to the frameworks enumerated. There exist other approaches which can be similarly
adapted such as Lyapunov-based-MPC \citep{de2008lyapunov}, predictive control using artificial reference \citep{Krupa10886854} or ANN-based control design with ellipsoidal PI set \citep{markolf2023tailored}. 

\thinh{Regarding the limitations of the proposed framework, although enumerating the members of the union of polytopes in \eqref{eq:cellenum} comprises elementary operations, the construction of the set may become intractable with respect to the number of polytopes created, or the size of the network. }
This problem naturally suggests a future direction of characterizing the invariance directly from the network parameters (weights and biases) instead of pre-processing them into interpretable geometric objects. \thinh{Finally, as the system is linear in the flat output space, we aim to further enhance the design by incorporating established robustification techniques from linear control theory to account for disturbances \citep{fortuna2021optimal}.}
\thinh{Looking ahead, we also aim to explore the use of ANN as a general tool for representing complex nonlinear mappings and enforcing safety constraints across different applications \citep{gao20253d,gao2025three,gao2025encrypt}.}

\section{Conclusion and outlook}
\label{sec:concl}

In this work, we proposed an offline procedure to find a positive invariant (PI) set associated with a Control Lyapunov Function (CLF) via neural network-based approximation for a class of nonlinear differentially flat systems.
More specifically, with the property of feedback linearizability, a single-input flat system can be transformed into a linear controllable dynamics at the price of distorting the feasible set. The distorted feasible set is then characterized by approximating the linearizing mapping with an artificial neural network (ANN) and represented as a union of polytopes. \thinh{With the property of a universal approximator of ANN and the piecewise-affine representation, such characterization provides a tight representation of the constraint admissible region.}
\thinh{Together with the linear dynamics in the new coordinates, such a polytopic encoding then gives rise to the construction of an ellipsoidal PI set and a CLF via established
methods, available in the linear control framework}. As a result, these stability and constraint-satisfying ingredients can be employed in various control design frameworks to obtain effective controllers with different levels of online complexity and performance. Simulation results are provided to demonstrate the applicability of the setting in control synthesis for nonlinear systems. 
\thinh{In a broader sense, the obtained results can be interpreted through the complementary roles of flatness and neural approximation. Flatness provides a linear structure that simplifies the nonlinear dynamics, while the ANN offers a computationally efficient parameterization of the associated nonlinear constraints. Together, these elements enable an effective characterization of the PI with reduced conservatism.}
For future directions, we anticipate characterizing the invariance directly from the network's parameters, bypassing the complex union of polytopes. 
\thinh{A direct comparison with other stability verification methods is also planned as a subsequent step.}
Furthermore, we expect that the methodology can be extended to applications in other frameworks involving ANN-based linearizing variable transformations, such as data-driven Koopman-based lifting linearization. 


\backmatter

\bmhead{Supplementary information}




\bmhead{Acknowledgements}
This work has benefited from a French government grant managed by the {Agence Nationale de la Recherche under the France 2030 program}, reference {ANR-23-IACL-0006}.

\bibliographystyle{sn-mathphys-num}



\end{document}